\documentclass[pra,aps,showpacs,twocolumn,twoside,superscriptaddress]{revtex4}

\usepackage{amsmath,amsfonts,amssymb,color,epsfig,graphics,graphicx,latexsym,revsymb,theorem,verbatim}

\newtheorem{definition}{Definition}
\newtheorem{proposition}[definition]{Proposition}

\newtheorem{lemma}[definition]{Lemma}

\def\squareforqed{\hbox{\rlap{$\sqcap$}$\sqcup$}}
\def\qed{\ifmmode\squareforqed\else{\unskip\nobreak\hfil
\penalty50\hskip1em\null\nobreak\hfil\squareforqed
\parfillskip=0pt\finalhyphendemerits=0\endgraf}\fi}
\def\endenv{\ifmmode\;\else{\unskip\nobreak\hfil
\penalty50\hskip1em\null\nobreak\hfil\;
\parfillskip=0pt\finalhyphendemerits=0\endgraf}\fi}
\newenvironment{proof}{\noindent \textbf{{Proof.~} }}{\qed}

\newcommand{\nc}{\newcommand}
\nc{\bra}[1]{\langle#1|} \nc{\ket}[1]{|#1\rangle} \nc{\proj}[1]{|
#1\rangle\!\langle #1 |} \nc{\ketbra}[2]{|#1\rangle\!\langle#2|}
\nc{\braket}[2]{\langle#1|#2\rangle} 
\nc{\norm}[1]{\lVert#1\rVert} \nc{\abs}[1]{|#1|}
\nc{\lar}{\leftarrow} \nc{\rar}{\rightarrow} \nc{\ox}{\otimes}
\nc{\op}[2]{|#1\rangle\!\langle#2|}
\nc{\ip}[2]{\langle#1|#2\rangle} \nc{\dg}{\dagger}
\nc{\I}{\mathbb{I}}

\nc{\cA}{{\cal A}} \nc{\cB}{{\cal B}} \nc{\cC}{{\cal C}}
\nc{\cD}{{\cal D}} \nc{\cE}{{\cal E}} \nc{\cF}{{\cal F}}
\nc{\cG}{{\cal G}} \nc{\cH}{{\cal H}} \nc{\cI}{{\cal I}}
\nc{\cJ}{{\cal J}} \nc{\cK}{{\cal K}} \nc{\cL}{{\cal L}}
\nc{\cM}{{\cal M}} \nc{\cN}{{\cal N}} \nc{\cO}{{\cal O}}
\nc{\cP}{{\cal P}} \nc{\cR}{{\cal R}} \nc{\cS}{{\cal S}}
\nc{\cT}{{\cal T}} \nc{\cX}{{\cal X}} \nc{\cZ}{{\cal Z}}

\begin{document}
\date{\today}
\title{Quantum Correlations in Large-Dimensional States of High Symmetry}
\author{Eric Chitambar}
\affiliation{The Perimeter Institute for Theoretical Physics, Waterloo, Ontario N2L 2Y5, Canada}

\begin{abstract}
In this article, we investigate how quantum correlations behave for the so-called Werner and pseudo-pure families of states.  The latter refers to states formed by mixing any pure state with the totally mixed state.  We derive closed expressions for the Quantum Discord (QD) and the Relative Entropy of Quantumness (REQ) for these families of states.  For Werner states, the classical correlations are seen to vanish in high dimensions while the amount of quantum correlations remain bounded and become independent of whether or not the the state is entangled.  For pseudo-pure states, nearly the opposite effect is observed with both the quantum and classical correlations growing without bound as the dimension increases and only as the system becomes more entangled.  Finally, we verify that pseudo-pure states satisfy the conjecture of [\textit{Phys. Rev. A} \textbf{84}, 052110 (2011)] which says that the Geometric Measure of Discord (GD) always upper bounds the squared Negativity of the state.
\end{abstract}


\maketitle

\section{Introduction}
\label{Sect:Intro}

Quantum and classical systems differ in many fascinating ways.  While quantum entanglement is one particular phenomenon unique to the quantum world, other non-classical features emerge when considering multipartite quantum systems.  Quite generally, correlations in a multipartite system exist when the state of some subsystem depends on the state of another.  Quantum systems, even those lacking entanglement, are capable of possessing correlations that cannot be simulated by classical physics.  \textit{Quantum correlations} (QC) beyond entanglement have recently received much attention as they have been discovered to play a prominent role in quantum information tasks such as thermodynamic work extraction \cite{Oppenheim-2002a}, the local broadcasting of correlations \cite{Piani-2008a}, the locking of classical correlations \cite{Datta-2009a, Boxio-2011a}, quantum state merging \cite{Cavalcanti-2011a, Madhok-2011a}, and the activation of entanglement in a measurement \cite{Streltsov-2011a, Piani-2011a}. 

Analogous to the case of entanglement measures, there exists no single quantifier of QC that adequately captures the ``quantumness'' in a given system.  Instead, the measure one uses to discuss correlations becomes relative to the task under investigation.  Conversely, any practical measure of QC ought to have some precise correspondence to a physical process, in addition to satisfying a number of other necessary properties \cite{Brodutch-2011a}.  In recent years, there have been various proposed measures for quantum correlations and a nice review of them can be found in Refs. \cite{Lang-2011a, Modi-2011a}.  

However despite the wealth of research conducted on the theory of QC \cite{Ollivier-2001a, Henderson-2001a, Oppenheim-2002a, Groisman-2005a, Horodecki-2005a, Luo-2008a, Piani-2008a, Wu-2009a, Ferraro-2010a, Dakic-2010a, Brodutch-2010a, Cavalcanti-2011a, Madhok-2011a, Boxio-2011a, Brodutch-2011b, Al-Qasimi-2011a, Chen-2011a, Cornelio-2011a}, very few calculations have been performed.  One reason for this is that the computations reduce to inherently difficult optimization problems.  Consequently, a general sense of how QC behave throughout state space is still lacking.  Furthermore, since the bulk of calculations already conducted are limited to qubit systems \cite{Luo-2008b, Ali-2010a, Yu-2011a, Vinjanampathy-2012a}, little is known about quantum correlations in larger dimensions or what new features emerge beyond qubits.

In this article we seek to shed light on some of these issues by computing the quantum correlations for the so-called Werner states \cite{Werner-1989a} and pseudo-pure states of $d\otimes d$ systems.  Werner states take the form
\begin{equation*}
\rho_W=\frac{2(1-\lambda)}{d(d+1)}\Pi^++\frac{2\lambda}{d(d-1)}\Pi^-
\end{equation*}
where $0\leq\lambda\leq 1$, and  $\Pi^+$ and $\Pi^-$ are projectors onto the symmetric and anti-symmetric subspaces of $\mathbb{C}^d\otimes\mathbb{C}^d$ respectively.  These states were originally introduced in the study of Bell inequalities as their measurement statistics can be described by a local hidden variable model \cite{Werner-1989a}.  Since then, they have been used to study various other aspects of entanglement theory \cite{Horodecki-2009a}.  The pseudo-pure (PP) states refer to states having the form
\begin{equation*}
\rho_{PP}=\alpha\op{\psi}{\psi}+\frac{1-\alpha}{d^2-1}(\mathbb{I}-\op{\psi}{\psi})
\end{equation*}
where $\ket{\psi}$ is an arbitrary pure state.  Physically we can think of pseudo-pure states as the mixing of some original pure state $\ket{\psi}$ with ``white noise'' represented by the completely mixed state $\mathbb{I}/d^2$.  PP states have be studied as a possible resource for NMR quantum computing \cite{Gershenfeld-1997a, Cory-1997a}, and they have also been proven useful for quantum computing without entanglement \cite{Biham-2004a}.  This latter result strongly motivates the work of this article since it is conjectured that quantum correlations beyond entanglement may be responsible (at least partially) for computational speed-ups in quantum computers \cite{Lanyon-2008a, Datta-2011a}; thus investigating quantum correlations in pseudo-pure states is a highly relevant project.  Note that when $\ket{\psi}$ is maximally entangled, $\rho$ becomes a so-called isotropic state, which represents another important class of states in entanglement theory \cite{Horodecki-2009a}.

Before presenting our specific findings for these family of states, we will first briefly review in Section I three important measures of quantum correlations: Quantum Discord (QD), Relative Entropy of Quantumness (REQ), and the Geometric Measure of Discord (GD).  In Section II we derive analytic formulas for the quantum correlations of Werner and PP families of states and study their behavior in high dimensions.  It is found that for PP states, the discord strictly decreases as more white noise is added.  We then consider a conjecture by Girolami and Adesso stating that the Geometric Measure of Discord is always an upper bound of the squared Negativity: $\mathcal{D}_G^{\leftarrow}(\rho)\geq \mathcal{N}^2(\rho)$ \cite{Girolami-2011a}.  While this has been proven for pure, Werner and isotropic states, here we show that it holds true for pseudo-pure states as well.  Concluding remarks are given in Sect. \ref{Sect:Conclusion}.

\section{Measures of Quantumness}
\label{Sect:QDREQ}
\subsection{Quantum Discord}
The definition of discord can be viewed as the mismatch in representations of mutual information that emerges when transitioning from classical to quantum information theory \cite{Ollivier-2001a}.  Classically, we consider random variables $A$ and $B$ and use the Shannon entropy $H(\cdot)$ as a measure of uncertainty in the given variable.  The mutual information between $A$ and $B$ is given by
\begin{align}
I(A:B)&=H(A)+H(B)-H(A,B)\notag\\
&=H(A)-H(A|B).
\end{align}
Here, $H(A|B)=\sum_b p(b)H(A|B=b)$ is the conditional entropy and indicates the expected uncertainty in $A$ that remains after learning $B$.  

In quantum information theory, the random variables become density operators with the joint ``variable'' being $\rho^{AB}$ and the marginals being the reduced states $\rho^A=tr_B(\rho^{AB})$ and $\rho^B=tr_A(\rho^{AB})$.  The uncertainty in a state $\rho$ is measured by the von Neumann entropy $S(\rho)=-tr(\rho\log\rho)$, and the quantum mutual information is defined as
\begin{align}
I(\rho^{AB})=S(\rho^A)+S(\rho^B)-S(\rho^{AB}).
\end{align}
Besides being a direct correspondence to the classical mutual information, quantum mutual information has a similar operational interpretation as the total correlations, both quantum and classical, that are present in the state $\rho^{AB}$ \cite{Groisman-2005a}.  

However, generalizing the classical conditional entropy becomes problematic since the expected uncertainty of Alice's system after Bob measures depends on the choice of measurement Bob makes.  More precisely, after Bob performs some measurement given by the POVM elements $\{E_i\}$ for which $\sum_iE_i=\mathbb{I}$, the average uncertainty in Alice's system is $S(A|\{E_i\}):=\sum_{i}p_iS(\rho^A_i)$ with $p_i=tr(E_i\rho^{AB})$ and $\rho^A_i=tr_B(E_i\rho^{AB})/p_i$.  The quantity $S(A)-S(A|\{E_i\})$ can then be interpreted as the amount of shared classical information between Alice and Bob with respect to the measurement $\{E_i\}$, as this is the average reduction in the uncertainty of Alice's system after Bob obtains the classical outcome of his measurement.  The total Classical Correlations (CC) in the state $\rho^{AB}$ is then obtained by maximizing over all possible POVMs by Bob \cite{Henderson-2001a}: 
\begin{equation}
\label{Eq:CC}
\mathcal{C}^\leftarrow(\rho^{AB})=\max_{\{E_i\}}[S(\rho^A)-S(A|\{E_i\})],
\end{equation}
and a rigorous justification for treating this as a measure of classical correlations has been given in terms of randomness distillation \cite{Devetak-2004a}.  A natural measure for quantum correlations is then the total correlations, $I(\rho^{AB})$, less the classical correlations, $\mathcal{C}^\leftarrow(\rho^{AB})$.  This intuition leads to the definition of \textit{Quantum Discord} (QD):
\begin{align}
\label{Eq:QD}
\mathcal{D}^\leftarrow(\rho^{AB})&=I(\rho^{AB})-\mathcal{C}^\leftarrow(\rho^{AB})\notag\\
&=S(\rho^B)-S(\rho^{AB})+\min_{\{E_i\}}\sum_{i}p_iS(\rho^A_i)
\end{align}
Note that this quantity is inherently asymmetric, an issue addressed in Ref. \cite{Cavalcanti-2011a}.  In this article, we will always assume that the discord is being considered with respect to a measurement on Bob's system, although the symmetry of the states we study will make this distinction unnecessary.  Also, due to the concavity of the von Neumann entropy, we can assume that each of the POVM elements $E_i$ are rank one.

\subsection{Relative Entropy of Quantumness (REQ)}
Another prominent measure of QC is the quantum \textit{Relative Entropy of Quantumness} (REQ) (defined in Eq. \eqref{Eq:REQ}) \cite{Groisman-2007a, Saitoh-2008a, Luo-2008a, Modi-2010a}.  REQ is analogous to the relative entropy of entanglement in that the latter measures the ``distance'' between a state and the set of separable states while the former measures the distance between a state and the set of classical states.  In the bipartite case, a \textit{quantum-classical} state takes the form
\begin{equation}
\label{Eq:FC}
\rho=\sum_{i}p_{i}\rho_i\otimes\op{\eta_i}{\eta_i}
\end{equation} where the $\{\ket{\eta_i}\}$ form an orthonormal basis for Bob's system and the $\rho_i$ are density matrices on Alice's.  If, in addition, all of the $\rho_i$ commute, than $\rho$ is said to be \textit{fully classical}.  This is because $\rho$ can then be regarded as a distribution of states in some two-party classical system.  

Let $\mathcal{Q}c$ denote the set of quantum-classical states and $\mathcal{F}c$ the set of fully classical states.  Then we can define measures of quantumness which, roughly speaking, quantifies how far a given state is from one of these chosen sets.   More precisely, recall that for two quantum states $\rho$ and $\sigma$, the quantum relative entropy of $\rho$ relative to $\sigma$ is defined by $S(\rho||\sigma)=-S(\rho)-tr\rho\log\sigma$ \cite{Ohya-1993a}.  Based on this notion of ``closeness'' between two states, we define two versions of a Relative Entropy of Quantumness:
\begin{align}
\label{Eq:REQ}
\mathcal{Q}^{\leftarrow}(\rho^{AB})&=\min_{\sigma\in\mathcal{Q}c}S(\rho^{AB}||\sigma)\notag\\
\mathcal{Q}(\rho^{AB})&=\min_{\sigma\in\mathcal{F}c}S(\rho^{AB}||\sigma).
\end{align}
The first represents an asymmetric measure just like $\mathcal{D}^{\leftarrow}$ while the second is meant to capture the full quantum correlations of the state from both sides.  It is the full QC version that is typically referred to as REQ, but the one-sided version has also been studied in the literature \cite{Modi-2010a}.  

Concerning the relationship between QD and REQ, it is not difficult to show \cite{Modi-2010a} that $\mathcal{D}^{\leftarrow}\geq\mathcal{Q}^{\leftarrow}$ with equality being achieved when the optimal measurement on Bob's side consists of projecting into a local eigenbasis.  Furthermore, when equality does hold and the conditional post-measurement states of Alice commute after Bob projects in an eigenbasis, we have that $\mathcal{D}^{\leftarrow}=\mathcal{Q}^{\leftarrow}=\mathcal{Q}$.  This is because for such a state $\rho$, the closest $\sigma\in\mathcal{Q}c$ will also belong to $\mathcal{F}c$; it will simply be $\rho$ dephased by Alice and Bob in a local eigenbasis.  Note that this also implies equality with the so-called measurement-induced disturbance (MID) measure of quantum correlations $\mathcal{D}_{MID}$ \cite{Luo-2008a}.
\begin{proposition}
\cite{Modi-2010a, Modi-2011a} If $\mathcal{D}^{\leftarrow}(\rho)$ is obtained by a projection into a local eigenbasis of Bob such that the conditional post-measurement states of Alice commute, then
\begin{equation}
\mathcal{D}^{\leftarrow}(\rho)=\mathcal{D}_{MID}(\rho)=\mathcal{Q}^{\leftarrow}(\rho)=\mathcal{Q}(\rho).
\end{equation}
\end{proposition} 
As we will prove below, both the Werner and PP families of states have this property, and thus their quantum correlations coincide with respect to these measures.    

\subsection{Geometric Measure of Discord}

The \textit{Geometric Measure of Discord} (GD) was originally introduced in Ref. \cite{Dakic-2010a} as a measure of quantum correlations, and recently, a physical interpretation has been given in terms of remote state preparation \cite{Dakic-2012a}.  As the name suggests, GD captures a distance-based notion of quantum correlations much like REQ.  For a state $\rho$, we define
\begin{equation}
\label{Eq:GQD}
\mathcal{D}^{\leftarrow}_G(\rho)=\frac{d}{d-1}\min_{\sigma\in\mathcal{Q}c}||\rho-\sigma||_2.
\end{equation}
Here $||\cdot||_2$ denotes the Hilbert-Schmidt norm; i.e. $||A||_2=\sqrt{tr[A^\dagger A]}$.  The above definition includes a normalization factor such that $0\leq\mathcal{D}^{\leftarrow}_G(\rho)\leq 1$, which is not its standard presentation.  The reason for normalizing is that the Hilbert-Schmidt norm is not stable under embedding into higher dimensions.  For instance $||\rho\otimes\mathbb{I}/n||_2=\sqrt{n}||\rho||_2$.  Thus without a normalization factor, it becomes difficult to interpret the numerical value of $\mathcal{D}^{\leftarrow}_G$ in high dimensions.  We also note that a symmetric version of $\mathcal{D}^{\leftarrow}_G$ can be introduced just like REQ by minimizing over the set $\mathcal{F}c$.  However, we will not consider this here.

Progress in computing $\mathcal{D}^{\leftarrow}_G(\rho)$ was made by Luo and Fu who showed that $\mathcal{D}^{\leftarrow}_G(\rho)=\tfrac{d}{d-1}\min tr(\rho-\sum_k\tau_k)^2$ where $\tau_k=\bra{\eta_k^B}\rho\ket{\eta_k^B}$ and the minimization is taken over all local bases $\ket{\eta_k^B}$ on Bob's system \cite{Luo-2010a}.  From here, it is relatively straightforward to show \cite{Luo-2010a, Modi-2011a} that
\begin{equation}
\label{Eq:GD}
\mathcal{D}^{\leftarrow}_G(\rho^{AB})=\frac{d}{d-1}\min_{\ket{\eta_k^B}}[tr(\rho^2)-\sum_k\tau_k^2].
\end{equation}

\section{Quantum Correlations in Specific Families of States}

\label{Sect:States} 

\subsection{Werner States}
In a $d\otimes d$ system, the Werner states constitute the family of states which are invariant under conjugation by $U\otimes U$ for any unitary $U$, i.e. $U\otimes U\rho U^\dagger\otimes U^\dagger=\rho$.  To characterize the form of these states, we introduce the flip operator $\mathbb{F}$ on $\mathbb{C}^d\otimes\mathbb{C}^d$ (i.e. $\mathbb{F}\ket{\phi\theta}=\ket{\theta\phi}$ for all $\ket{\theta}$ and $\ket{\phi}$).  Then the projectors on the symmetric and anti-symmetric subspaces of $\mathbb{C}^d\otimes\mathbb{C}^d$ are $\Pi^+=(\mathbb{I}+\mathbb{F})/2$ and $\Pi^-=(\mathbb{I}-\mathbb{F})/2$ respectively.  A key property of $\Pi^+$ and $\Pi^-$ is that they take the same form regardless of the basis used to represent them.  In other words, for any orthonormal basis $\{\ket{\phi_i}\}_{i=1...d}$ we have $\mathbb{F}=\sum_{i,j=1}^d\op{\phi_j\phi_i}{\phi_i\phi_j}$ and therefore
\begin{align}
\label{Eq:symproj}
\Pi^+&=\sum_{i=1}^d\op{\phi_i\phi_i}{\phi_i\phi_i}+\sum_{i<j}^{d(d-1)/2}\op{\Psi^+_{ij}}{\Psi^+_{ij}}\notag\\
\Pi^-&=\sum_{i<j}^{d(d-1)/2}\op{\Psi^-_{ij}}{\Psi^-_{ij}}
\end{align}
where $\ket{\Psi^+_{ij}}=\sqrt{\frac{1}{2}}\left(\ket{\phi_i\phi_j}+\ket{\phi_j\phi_i}\right)$, and $\ket{\Psi^-_{ij}}=\sqrt{\frac{1}{2}}\left(\ket{\phi_i\phi_j}-\ket{\phi_j\phi_i}\right)$.  A general Werner state is
\begin{equation}
\label{Eq:Werner}
\rho_{W}=\frac{2(1-\lambda)}{d(d+1)}\Pi^++\frac{2\lambda}{d(d-1)}\Pi^-
\end{equation}
with $0\leq\lambda\leq 1$.  Note that $\lambda=tr(\rho\Pi^-)$ in Eq. \eqref{Eq:Werner}, and it is known that $\rho$ is separable if and only if $\lambda\leq 1/2$ \cite{Werner-1989a}.

For the mutual information of a Werner state, we first have
\begin{equation}
\label{Eq:JointWerner}
S(\rho_W)=-\lambda\log\frac{2\lambda}{d(d-1)}-(1-\lambda)\log\frac{2(1-\lambda)}{d(d+1)}
\end{equation}
while both $\rho^A$ and $\rho^B$ are totally mixed with $S(\rho^A)=S(\rho^B)=\log d$.  Therefore,
\begin{align}
\label{Eq:MIWerner}
I(\rho_W)&=2\log d+\lambda\log\frac{2\lambda}{d(d-1)}+(1-\lambda)\log\frac{2(1-\lambda)}{d(d+1)}\notag\\
&=\log 2d+\lambda\log\frac{\lambda}{d-1}+(1-\lambda)\log\frac{1-\lambda}{d+1}.
\end{align}

To compute the minimum conditional entropy of Alice for a Werner state when Bob performs a measurement, we use the fact that this can be obtained by a rank one POVM measurement.  This can be represented by a set of operators $\{\op{\eta_k}{\tilde{\eta}_k}\}_{k=1,...}$ where $\ket{\tilde{\eta}_k}$ is some state with $b_k=\ip{\tilde{\eta}_k}{\tilde{\eta}_k}$ and $\ket{\eta_k}=b_k^{-1/2}\ket{\tilde{\eta}_k}$.  For each $\ket{\tilde{\eta}_k}$, let $\{\ket{k^\perp_j}\}_{j=1,...,d-1}$ be an orthonormal basis that spans the orthogonal complement to $span\{\ket{\tilde{\eta}_k}\}$.  We will use this basis to denote the Werner state:
\begin{align}
\label{WS2}
&\rho_W=\tfrac{2(1-\lambda)}{d(d+1)}\op{\eta_k\eta_k}{\eta_k\eta_k}\notag\\
&+\tfrac{1-\lambda}{d(d+1)}\sum_{j=1}^{d-1}[\ket{\eta_k k^\perp_j}+\ket{k^\perp_j\eta_k}][\bra{\eta_k k^\perp_j}+\bra{k^\perp_j \eta_k}]\notag\\
&+\tfrac{\lambda}{d(d-1)}\sum_{j=1}^{d-1}[\ket{\eta_k k^\perp_j}-\ket{k^\perp_j\eta_k}][\bra{\eta_k k^\perp_j}-\bra{k^\perp_j\eta_k}]+T
\end{align}
where $T$ is some operator orthogonal to $I\otimes\op{\tilde{\eta}_k}{\tilde{\eta}_k}$.  We then see that the (unnormalized) post-measurement state is
\begin{align}
\tfrac{2(1-\lambda)b_k}{d(d+1)}\op{\eta_k\eta_k}{\eta_k\eta_k}+\tfrac{(d-1+2\lambda)b_k}{d(d-1)(d+1)}\sum_{j=1}^{d-1}\op{k^\perp_j\eta_k}{k^\perp_j\eta_k}.
\end{align}
Thus we have $p_k=b_k/d$ and
\begin{equation}
\rho^A_k=\tfrac{2(1-\lambda)}{d+1}\op{\eta_k}{\eta_k}+\tfrac{d-1+2\lambda}{(d-1)(d+1)}\sum_{j=1}^{d-1}\op{k^\perp_j}{k^\perp_j}.
\end{equation}
As this equation provides a diagonal form for $\rho^A_k$, we immediately obtain
\begin{align}
\label{Eq:postmeasureWerner}
S(&\rho^A_k)=-\tfrac{2(1-\lambda)}{d+1}\log\tfrac{2(1-\lambda)}{d+1}-\tfrac{d-1+2\lambda}{d+1}\log\tfrac{d-1+2\lambda}{d^2-1},
\end{align}
which does not depend on $k$ or the particular POVM.  Substituting into the definition of discord and performing some simplifications, we find that
\begin{align}
\label{Eq:Wernerstatediscord}
\mathcal{D}^\leftarrow(\rho_{W})&=\log(d+1)+\lambda\log\tfrac{\lambda}{d-1}+(1-\lambda)\log\tfrac{1-\lambda}{d+1}\notag\\
&-\tfrac{2(1-\lambda)}{d+1}\log (1-\lambda)-\tfrac{d-1+2\lambda}{d+1}\log\tfrac{d-1+2\lambda}{2(d-1)}.
\end{align}
Since the conditional states of Alice are diagonal in the same basis after Bob projects into any orthonormal basis, the assumption of Proposition 1 is fulfilled, and thus all measures of quantum correlations for Werner states coincide.

\begin{figure}[t]
\includegraphics[scale=.5]{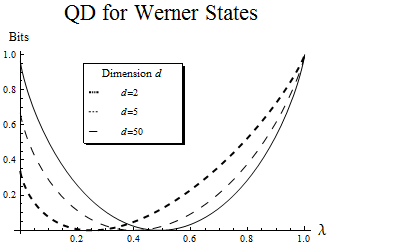}
\caption{\label{QDWerner}
The quantum discord in Werner states of various dimensions.} 
\end{figure}

\begin{figure}[t]
\includegraphics[scale=.5]{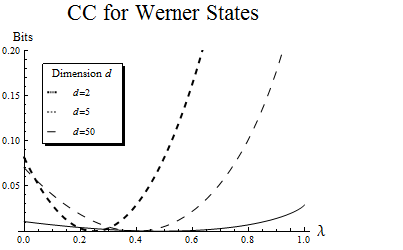}
\caption{\label{CCWerner}
The classical correlations in Werner states of various dimensions.} 
\end{figure}

The discord and classical correlations (Eq. \eqref{Eq:CC}) are plotted in Figs. \ref{QDWerner} and \ref{CCWerner} respectively for different dimensions.  The only fully classical state in the Werner family is the totally mixed state $\mathbb{I}/d^2$.  From Fig. \ref{QDWerner}, we can see that for higher dimensions the discord approaches complete symmetry about the point $\lambda=1/2$, which corresponds to the border of entangled versus non-entangled states.  We can explicitly see the symmetry as well as the boundedness by expanding Eq. \eqref{Eq:Wernerstatediscord} for large $d$:
\begin{equation}
\mathcal{D}^\leftarrow(\rho_{W})\approx 1-H(\lambda), \qquad d>>1
\end{equation}
where $H(\lambda)=-\lambda\log\lambda-(1-\lambda)\log(1-\lambda)$ is the binary Shannon entropy.  

The symmetry of the discord in higher dimensions means that for any entangled Werner state possessing less than 1 bit of QD, there exists a separable state which possesses the same amount of discord.   Recall that any separable state can be created by local operations and classical communication (LOCC) \cite{Werner-1989a}.  These results imply that almost any amount of discord achievable in the Werner family of states can be obtained by LOCC.  The one exception here is the 1 bit maximum which requires entanglement to be reached exactly (i.e. $\lambda=1$).

In Fig. \ref{CCWerner}, we can see that as $d$ grows large, the classical correlations go to zero and nearly all the correlations become quantum, for both entangled and separable states.  This can be confirmed by the fact that asymptotically, the mutual information $I(\rho_W)$ also becomes $1-H(\lambda)$, and so $\mathcal{D}^\leftarrow(\rho_W)\rightarrow I(\rho_W)$ as $d\to\infty$.  However, despite this general convergence, for any dimension $d$, the classical correlations exactly vanish for the one and only value of $\lambda=(d-1)/(2d)$.  This is the same parameter value for which QD vanishes and corresponds to the totally mixed state.  Thus, for any $d'>d$, the classical correlations will be larger at the point $\lambda=(d-1)/(2d)$ in the $d'\times d'$-dimensional Werner state.

We next compare QD to the entanglement of the state as measured by the Entanglement of Formation (Eof).  We note that similar comparisons between discord and entanglement have been made elsewhere \cite{Luo-2008a, Ali-2010a, Al-Qasimi-2011a}, but these studies were limited only to two qubits.  EoF for Werner states has a closed form expression given by \cite{Vollbrecht-2001a}:
\begin{align}
E&_f(\rho)=1-(1/2-\sqrt{\lambda(1-\lambda)})\log[1-2\sqrt{\lambda(1-\lambda)}]\notag\\
&-(1/2+\sqrt{\lambda(1-\lambda)})\log[1+2\sqrt{\lambda(1-\lambda)}]\;\;\;(\lambda>1/2).
\end{align}
Note that this quantity is independent of the dimension.  In Fig. \ref{EoFWerner}, we simultaneously plot QD and EoF for dimensions $d=2$ and $d=50$.  It can be seen that the EoF becomes a general upper bound for the discord as $d$ grows large.  However, careful analysis shows that for any finite $d$, there will always exist a range $[1/2,1/2+\epsilon(d))$ for which QC $>$ EOF.

\begin{figure}[t]
\includegraphics[scale=.5]{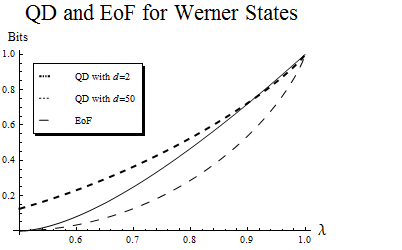}
\caption{\label{EoFWerner}
The quantum discord and Entanglement of Formation (EoF) of Werner states in different dimensions $d$.} 
\end{figure}

\subsection{Pseudo-Pure States}

We next consider the quantumness of pseudo-pure states.  These are states of the form:
\begin{align}
\rho_{PP}=\alpha\op{\psi}{\psi}+\beta(\mathbb{I}-\op{\psi}{\psi})
\end{align}
where $\ket{\psi}=\sum_{i=1}^d u_i\ket{ii}$ with $u_1\geq u_2\geq...\geq u_d\geq 0$ and $\beta=(1-\alpha)/(d^2-1)$.  From Ref. \cite{Vidal-1999b}, we know that $\rho$ is separable iff $\lambda\leq\frac{1+u_1u_2}{1+u_1u_2d^2}$.  We consider a rank one POVM $\{\op{\eta_k}{\tilde{\eta}_k}\}_{k=1,...}$ on Bob's system with $\ket{\tilde{\eta}_k}=\sum_{i=1}^dc_{ki}\ket{i}$.    Define $\ket{\psi_k}:=\frac{1}{\sqrt{p_k}}\ip{\tilde{\eta}_k}{\psi}$ where $p_k=tr[\ip{\tilde{\eta}_k}{\psi}\ip{\psi}{\tilde{\eta}_k}]=\sum_{i=1}^d|c_{ki}|^2u_i^2$.  Then upon outcome $k$, Alice's post-measurement state is
\begin{align}
\rho^A_k=&\frac{\sum_{i=1}^d|c_{ki}|^2[u_i^2(\alpha-\beta)+\beta]}{q_k}\op{\psi_k}{\psi_k}\notag\\
&+\frac{\sum_{i=1}^d|c_{ki}|^2\beta}{q_k}(\mathbb{I}-\op{\psi_k}{\psi_k})
\end{align}
in which $q_k=\sum_{i=1}^d|c_{ki}|^2[u_i^2(\alpha-\beta)+\beta d]$.  The average residual entropy is given by 
\begin{align}
&\sum_k p_kS(\rho^A_k)=\notag\\
&-\sum_k\sum_{i=1}^d|c_{ki}|^2[u_i^2(\alpha-\beta)+\beta]\log\tfrac{\sum_{i=1}^d|c_{ki}|^2[u_i^2(\alpha-\beta)+\beta]}{\sum_{i=1}^d|c_{ki}|^2[u_i^2(\alpha-\beta)+\beta d]}\notag\\
&-\sum_k\sum_{i=1}^d|c_{ki}|^2\beta(d-1)\log\tfrac{\sum_{i=1}^d|c_{ki}|^2\beta}{\sum_{i=1}^d|c_{ki}|^2[u_i^2(\alpha-\beta)+\beta d]}.
\end{align}
We next apply the log-sum inequality \cite{Cover-2006a} and use the fact that $\sum_k|c_{ki}|^2=1$ to obtain
\begin{align}
\sum_k q_kS(\rho^A_k)\geq &-\sum_{i=1}^d[u_i^2(\alpha-\beta)+\beta]\log\tfrac{u_i^2(\alpha-\beta)+\beta}{u_i^2(\alpha-\beta)+\beta d}\notag\\
&-\sum_{i=1}^d\beta(d-1)\log\tfrac{\beta}{u_i^2(\alpha-\beta)+\beta d}.
\end{align}
This lower bound is saturated by an orthogonal projective measurement by Bob in a local eigenbasis.  In fact, this is the only such measurement that obtains the lower bound.  To compute the discord, we use the facts that
\[S(\rho_{PP})=-\alpha\log\alpha-\beta(d^2-1)\log\beta\]
and
\[S(\rho^A)=-\sum_i[d\beta+u_i^2(\alpha-\beta)]\log[d\beta+u_i^2(\alpha-\beta)]\]
with the normalization $\beta=(1-\alpha)/(d^2-1)$ to find
\begin{align}
\label{Eq:DisPP}
\mathcal{D}^{\leftarrow}&(\rho_{PP})=\alpha\log\alpha+\tfrac{1-\alpha}{d+1}\log\tfrac{1-\alpha}{d^2-1}\notag\\
&-\sum_{i=1}^d[\tfrac{(1-\alpha)-u_i^2(1-d^2\alpha)}{d^2-1}]\log[\tfrac{(1-\alpha)-u_i^2(1-d^2\alpha)}{d^2-1}].
\end{align}
Like Werner states, the conditional post-measurement states of Alice are all diagonal in a local eigenbasis.  Therefore by Proposition 1 the quantum correlations of PP states are equal for the various measures.  

From Eq. \eqref{Eq:DisPP} we see that that in the large limit of $d$, the discord approaches
\begin{equation}
\mathcal{D}^{\leftarrow}(\rho_{PP})\approx \alpha \cdot S(tr_B\op{\psi}{\psi}), \qquad d>>1.
\end{equation}
In other words, the discord grows linearly with $\lambda$ with a slope equaling the reduced state entropy of $\ket{\psi}$.

From Eq. \eqref{Eq:DisPP}, we can deduce the following lemma.
\begin{lemma}
For a fixed state $\ket{\psi}$ the discord of its corresponding PP state is convex with respect to the mixing parameter $\alpha$.  Consequently, the discord of $\alpha\op{\psi}{\psi}+\beta(\mathbb{I}-\op{\psi}{\psi})$ monotonically decreases with $\alpha$.
\end{lemma}
\begin{proof}
Computing the second derivative of Eq. \eqref{Eq:DisPP} with respect to $\alpha$ gives
\begin{align}
&\tfrac{1}{\alpha}+\tfrac{1}{(1-\alpha)(1+d)}-\sum_{i=1}^d\tfrac{(d^2u_i^2-1)^2}{(d^2-1)(1-u_i^2+\alpha(d^2u_i^2-1))}\notag\\
&\geq\tfrac{1}{\alpha}+\tfrac{1}{(1-\alpha)(1+d)}-\sum_{i=1}^d\tfrac{(d^2u_i^2-1)^2}{(d^2-1)(d^2u_i^2-1)\alpha}\notag\\
&=\tfrac{1}{(1-\alpha)(1+d)}>0.
\end{align}
\end{proof}

To study how quantum correlations behave with PP states, we consider the isotropic state which is a PP state with $\ket{\psi}$ being maximally entangled (i.e. $u_i=\sqrt{1/d}, \forall i$).  Both QD and CC for isotropic states are depicted in Figs. \ref{Fig:QDIso} and \ref{Fig:CCIso} respectively.  

Unlike the Werner states, the classical correlations do not vanish in large dimensions, as evident from Fig. \ref{Fig:CCIso}.  
One might wonder how the quantum and classical correlations compare to one another.  It can be shown that QD $\geq$ CC in general, although the exact expression for CC is a bit messy.  However, for large $d$, the classical correlations behave as
\begin{equation}
\mathcal{C}^{\leftarrow}(\rho)\approx\mathcal{D}^{\leftarrow}(\rho)-H(\alpha), \qquad d>>1.
\end{equation}
In Fig. \ref{Fig:QC-CCIso} we plot the difference.  

\begin{figure}[b]
\includegraphics[scale=.5]{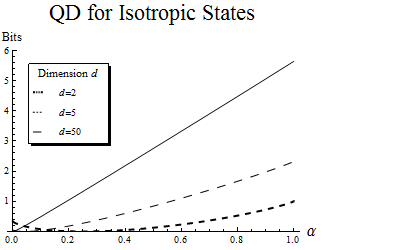}
\caption{\label{Fig:QDIso}
The quantum discord in isotropic states of various dimensions.} 
\end{figure}

\begin{figure}[t]
\includegraphics[scale=.5]{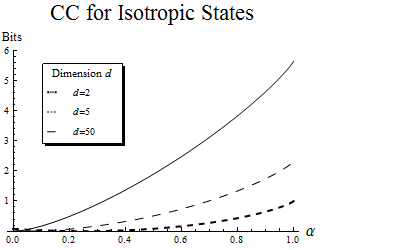}
\caption{\label{Fig:CCIso}
The classical correlations in isotropic states of various dimensions.} 
\end{figure}

\begin{figure}[t]
\includegraphics[scale=.5]{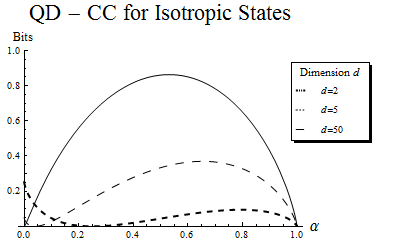}
\caption{\label{Fig:QC-CCIso}
The difference in discord versus classical correlations for isotropic states.  For large $d$, the difference approach $H(\alpha)$.} 
\end{figure}

\bigskip

\noindent\textbf{The GD and Negativity of PP States}

Now we proceed to compute the Geometric Measure of Discord for PP states.  For the basis $\ket{\eta_k}$ on Bob's side, let $p_k=tr[\ip{\eta_k}{\psi}\ip{\psi}{\eta_k}]$.  Then we have
\begin{equation}
\bra{\eta_k}\rho\ket{\eta_k}=[\beta +p_k(\alpha-\beta)]\op{\psi_k}{\psi_k}+\beta(\mathbb{I}-\op{\psi_k}{\psi_k})
\end{equation}
and so 
\begin{align}
tr[\bra{\eta_k}\rho\ket{\eta_k}^2]=f(p_k)+(d-1)\beta^2
\end{align}
where $f(x)=[\beta+x(\alpha-\beta)]^2$.  The function $f(x)$ is convex, and so $\sum_k f(p_k)\leq\sum_k f(q_k)$ for any probability distribution $(q_k)$ that  majorizes  $(p_k)$: $(p_k)\prec(q_k)$ \cite{Bhatia-2000a}.  A well known majorization result says that $(p_k)\prec (u_k^2)$ since the $p_k$ are measurement statistics on the state $tr_A\op{\psi}{\psi}$ whose eigenvalues are $u_k^2$ \cite{Nielsen-2000a}.  Thus, 
\begin{align*}
\sum_ktr[\bra{\eta_k}\rho\ket{\eta_k}^2&\leq \sum_k [f(u_k^2)+(d-1)\beta^2]\\
&=\sum_k[\beta^2 d+2u_k\beta(\alpha-\beta)+u_k^4(\alpha-\beta)^2]\\
&=\beta^2 d^2+\alpha^2-\beta^2-2\sum_{i<j}u_i^2u_j^2(\alpha-\beta)^2
\end{align*}
Referring to Eq. \eqref{Eq:GD}, we see that
\begin{align}
\label{Eq:PPGD}
\mathcal{D}_G^{\leftarrow}(\rho_{PP})&=\frac{d}{d-1}2\sum_{i<j}u_i^2u_j^2(\alpha-\beta)^2\notag\\
&=\frac{2d(\alpha d^2-1)^2}{(d^2-1)^3}\sum_{i<j}u_i^2u_j^2\notag\\
&=\frac{(\alpha d^2-1)^2}{(d^2-1)^2}\mathcal{D}^{\leftarrow}_G(\psi)
\end{align}
where $\mathcal{D}^{\leftarrow}_G(\psi)=\frac{d}{d-1}(1-tr(\rho^A)^2)$ is the GD of the pure state $\ket{\psi}$ (equivalently the linear entropy).  Note that like the discord, GD is also convex with respect to the mixing parameter $\alpha$.

In Ref. \cite{Girolami-2011a} Girolami and Adesso studied the relationship between measures of entanglement and more general measures of quantumness.  We have already seen that there is no general hierarchy between EoF and QD.  However, Girolami and Adesso consider GD and the Negativity measure of entanglement \cite{Vidal-2002b}.  In its normalized form the Negativity of a state $\rho$ is given by $\mathcal{N}(\rho)=\tfrac{1}{d-1}(||\rho^{\Gamma}||_1-1)$.  Here $||A||_1=tr\sqrt{A^\dagger A}$ is the trace norm and $A^{\Gamma}$ is the partial transpose of $A$.  For density matrices, $||\rho^{\Gamma}||_1-1$ is computed by summing over the negative eigenvalues of $\rho^{\Gamma}$ and taking twice its absolute value.  Girolami and Adesso were able to prove that $\mathcal{N}^2\leq\mathcal{D}_G^{\leftarrow}$ for all pure states, two-qubit states, and the Werner and isotropic states.  The intuition is that if discord captures a greater amount of correlations than entanglement, there should be some way to quantify this relationship.  Thus, $\mathcal{N}^2\leq\mathcal{D}_G^{\leftarrow}$ seems like a reasonable conjecture for all states.

While the Negativity is a computable measure, for higher dimensional systems it generally does not have an analytic form.  However for PP states $\alpha\op{\psi}{\psi}+\beta(\mathbb{I}-\op{\psi}{\psi})$, the computation can be easily done.  For $\ket{\psi}=\sum_iu_i\op{ii}{ii}$, we have 
\begin{align}
\psi^{\Gamma}&=\sum_{i,j}u_iu_j\op{ij}{ji}\notag\\
&=\sum_{i}u_i^2\op{ii}{ii}+\sum_{i<j}u_iu_j[\Psi^+_{ij}-\Psi^-_{ij}]
\end{align}
where $\Psi^{\pm}_{ij}=\op{\Psi^{\pm}_{ij}}{\Psi^{\pm}_{ij}}$ and $\ket{\Psi^{\pm}_{ij}}=\sqrt{1/2}(\ket{ij}\pm\ket{ji})$. Thus,
\begin{align}
\rho_{PP}^{\Gamma}&=\sum_i[\beta+(\alpha-\beta)u_i^2]\op{ii}{ii}\notag\\
+&\sum_{i<j}[\beta+(\alpha-\beta)u_iu_j]\Psi^+_{ij}+\sum_{i<j}[\beta-(\alpha-\beta)u_iu_j]\Psi^{-}_{ij}.
\end{align}
Let $S_{-}$ be the set of $(i,j)$ such that $i<j$ and $\beta-(\alpha-\beta)u_iu_j<0$.  Then
\begin{equation}
\mathcal{N}(\rho_{PP})=\tfrac{2}{(d-1)^2(d+1)}\sum_{(i,j)\in S_{-}}[u_{i}u_j(\lambda d^2-1)-(1-\lambda)].
\end{equation}

To obtain a bound, we have
\begin{align}
\mathcal{N}^2(\rho_{PP})\leq [\tfrac{2(\lambda d^2-1)}{(d-1)(d^2-1)}\sum_{i<j} u_iu_j]^2=\frac{(\lambda d^2-1)^2}{(d^2-1)^2}\mathcal{N}^2(\psi).
\end{align}
Then using the fact that $\mathcal{D}_G(\psi)\geq \mathcal{N}(\psi)^2$ for pure states $\ket{\psi}$ \cite{Girolami-2011a}, we see that $\mathcal{D}_G^{\leftarrow}(\rho_{PP})\geq \mathcal{N}^2(\rho_{PP})$.

\section{Conclusion}
\label{Sect:Conclusion}
In this article we have explicitly computed the quantum correlations for Werner states and pseudo-pure states.  For these families of states, we have shown that the Quantum Discord and the Relative Entropy of Quantumness coincide.  Explicit formulas for the correlations are given by Eqs. \eqref{Eq:Wernerstatediscord} and \eqref{Eq:DisPP} for Werner states and pseudo-pure state respectively.  The Geometric Measure of Discord for PP states is given by Eq. \eqref{Eq:PPGD}.

We remark that this is one of the very few studies where the quantum correlations have actually been calculated.  Doing so has revealed an interesting variety of behavior.  For Werner states, the quantum correlations are always bounded by 1 while converging to the symmetric functional form of $1-H(\lambda)$ in higher dimensions.  On the other hand, for PP states, the quantum correlations behave like $\alpha H(\psi)$ for large $d$.  Physically, we see that by introducing randomness to an initially pure state can only decrease the quantum correlations.  Finally, we have shown that $\mathcal{D}^{\leftarrow}_G\geq \mathcal{N}^2$ for all PP states.  It is natural to wonder what other families of states satisfy this hierarchy, if not for all states.  We hope this research will provide a deeper insight into the structure of quantum correlations.

\noindent\textit{Note added:}  For two qubit systems, equations similar to \eqref{Eq:Wernerstatediscord} and \eqref{Eq:DisPP} for the isotropic state have been previously observed \cite{Piani-2011b}.

\begin{acknowledgments}
I'd like to thank Christian Weedbrook, Asma Al-Qasimi, Gerardo Adesso and Marco Piani for partaking in helpful discussions on this topic.  This work was supported by funding agencies including CIFAR, the CRC program, NSERC, and QuantumWorks.
\end{acknowledgments}

\bibliography{EricQuantumBib}

\end{document}